\documentclass[onecolumn,prl,superscriptaddress,notitlepage]{revtex4-1}

\usepackage{natbib,enumitem}
\usepackage{amsmath,amssymb,amsthm,mathrsfs,amsfonts}
\usepackage[latin9]{inputenc}
\usepackage{units}
\usepackage{graphicx}
\usepackage{esint}
\usepackage{braket}
\usepackage{multirow}
\usepackage{slashbox}
\usepackage{array}
\usepackage{lipsum}
\usepackage{color}

\makeatletter

\newcommand{\Yb}{\ensuremath{^{171}\mathrm{Yb}^+~}}
\newcommand{\Ba}{\ensuremath{^{138}\mathrm{Ba}^+~}}
\newcommand{\Ca}{\ensuremath{^{40}\mathrm{Ca}^+~}}

\newcommand{\avg}[1]{\ensuremath{\left\langle#1\right\rangle}}

\newtheorem{lemma}{Lemma}

\newtheorem{proposition}{Proposition}

\makeatother

\begin{document}
\title{Randomness expansion secured by quantum contextuality}

\author{Mark Um$^{1*}$, Qi Zhao$^{1*}$, Junhua Zhang$^{2,1}$, Pengfei Wang$^{1}$, Ye Wang$^{1}$, Mu Qiao$^{1}$, Hongyi Zhou$^{1}$, Xiongfeng Ma$^{1}$, and Kihwan Kim$^{1}$}

\affiliation{Center for Quantum Information, Institute for Interdisciplinary Information Sciences, Tsinghua University, Beijing 100084, P. R. China \\$^{2}$Shenzhen Institute for Quantum Science and Engineering, and Department of Physics, Southern University of Science and Technology, Shenzhen 518055, P. R. China\\$^{*}$These authors contributed equally to this work}

\begin{abstract}
The output randomness from a random number generator can be certified by observing the violation of quantum contextuality inequalities based on the Kochen-Specker theorem. Contextuality can be tested in a single quantum system, which significantly simplifies the experimental requirements to observe the violation comparing to the ones based on nonlocality tests. However, it is not yet resolved how to ensure compatibilities for sequential measurements that is required in contextuality tests. Here, we employ a modified Klyachko-Can-Binicio\u{g}lu-Shumovsky contextuality inequality, which can ease the strict compatibility requirement on measurements. On a trapped single \Ba ion system, we experimentally demonstrate violation of the contextuality inequality and realize self-testing quantum random number expansion by closing detection loopholes. We perform $1.29 \times 10^8$ trials of experiments and extract the randomness of $8.06 \times 10^5$ bits with a speed of 270 bits s$^{-1}$. Our demonstration paves the way for the practical high-speed spot-checking quantum random number expansion and other secure information processing applications.
\end{abstract}

\date{\today}
\maketitle

Randomness is a critical resource for information processing with applications ranging from computer simulations \cite{Coddington94} to cryptography \cite{Gisin02}. For cryptographic purposes, in particular, streams of random numbers should have good statistical behavior and unpredictability against adversaries \cite{Fiorentino07,Goldreich07}. In reality, random numbers produced by an algorithm or a classical chaotic process are intrinsically deterministic, thereby in principle allowing an adversary with the information of the device to find a pattern. On the other hand, the nature of quantum mechanics is fundamentally random, which, in this sense, provides a foundation for genuine randomness. Thanks to the unpredictable behavior of quantum mechanics, various quantum random number generators have been proposed and implemented \cite{Ma2016Quantum,Herrero2017Quantum,Liu2018High}. In practice, however, the security can be jeopardized if an adversary partially manipulates the devices or the devices are exposed to imperfection or malfunction. In order to address this realistic issue, the device-independent protocols have been proposed to guarantee the generated randomness without relying on detailed knowledge of uncharacterized devices \cite{Colbeck07,Pironio10,Colbeck2011private,Vazirani12,Pironio2013Security,Coudron13,Carl17,Chung2014Physical,Arnon2016,Ac2016Certified}.

The essence of device-independent randomness expansion lies in the fact that any violation of nonlocality inequalities \cite{bell1964einstein} shows unpredictability of measurement results. Recent security proofs show that randomness can be certified under the device-independent scenario by a class of Bell inequalities \cite{Colbeck07,Pironio10,Colbeck2011private,Vazirani12,Pironio2013Security,Coudron13,Carl17}. On the experimental side, the loophole-free violations of Bell's inequality have been demonstrated \cite{hensen2015Loophole,Shalm2015Strong,Giustina2015Significant}, which have been applied to generate random numbers \cite{bierhorst2018experimentally,liu2018device}. However, the randomness certification by the loophole-free Bell test is suffered from the low generation rate and requires high-fidelity entanglement sources. Moreover, it requires a large space separation between two detection sites to rule out the locality loophole, which is almost impossible to make the whole system compact. Till now, a strict and practical randomness expansion, where the output randomness is larger than input randomness, based on loophole-free Bell tests still has not been demonstrated and remained as an experimental challenge.

Similar to the Bell theorem, the Kochen-Specker theorem \cite{Bell66,KS67} states that quantum mechanics is contextual and cannot be fully explained by classical models, \emph{i.e.}, noncontextual hidden variables models that have definite predetermined values for measurement outcomes. Contextuality can be tested with a single system without entanglement by using the Klyachko-Can-Binicio\u{g}lu-Shumovsky (KCBS) inequality \cite{klyachko2008simple}, which can significantly reduce the experimental requirements comparing to the nonlocality test. Inequalities based on the Kochen-Specker theorem can provide alternatives for randomness certification, which has been studied in both theory and experiment \cite{Dongling12,UMark13,Carl17}. A contextuality test contains a set of contexts, which are composed of a certain number of compatible, $i.e.$, commuting in quantum mechanics, measurements. Note that the measurements in the nonlocality Bell test can also be regarded as compatible measurements. The randomness certification has been proven for the case with perfectly compatible measurements \cite{Carl17}. In reality, when the contextuality test is performed on a single party, it is difficult to establish the perfect compatibility between sequential measurements. Though a couple of experimental demonstration of randomness certification with the KCBS inequality have been reported \cite{Dongling12,UMark13}, the security of the scheme has not been fully resolved.

In this work, first, we experimentally demonstrate the violation of a modified KCBS inequality \cite{Gunhe10,Szangolies13}, which reveals quantum correlations without the requirement of the perfect compatibility on sequential measurements. Then we employ it for a spot-checking protocol of randomness expansion with exponential gain \cite{Carl17}, which is the first experimental demonstration of the strict randomness expansion. Our scheme is not a fully device-independent protocol, since it requires a few assumptions on the device, in particular, the assumption of approximate compatibilities on the measurement settings \cite{Herrero2017Quantum,lunghi2015self}. However, we do not need the perfect compatibility, since the imperfections in control and the disturbances from classical and quantum noisy-environment are characterized and compensated in the modified KCBS inequality. In this scenario, we can expand the randomness from the generated strings merely based on the experimental observed data that violate the modified KCBS inequality , which is in a self-testing manner \cite{Herrero2017Quantum,lunghi2015self}. We implement the protocol with a single trapped \Ba ion instead of a \Yb ion which was used for the previous demonstration \cite{UMark13} in order to fully address the experimental requirements in a modified KCBS inequality \cite{Gunhe10,Szangolies13}. The \Ba ion has long-lived states that can be used for the coherent shelving of a quantum state during the sequential measurements. We develop a narrow-line laser system that is stabilized to a high-finesse cavity to precisely manipulate the long-lived states and observe sufficient amount of violation for the randomness expansion with large enough number of trials. We perform $1.29 \times 10^8$ trials of experiments and extract the randomness of $8.06 \times 10^5$ bits with the speed of 270 bits s$^{-1}$. 

\section{Results}
\subsection*{Modified KCBS inequality}
In order to test contextuality, various inequalities have been proposed \cite{Cabello08,klyachko2008simple} and demonstrated in diverse physical systems, including trapped ion system \cite{Roos09,Xiang13}, photonic system \cite{Zeilinger11,xiao2018experimental}, and superconducting system \cite{Jerger16}. Among the contextuality inequalities, the KCBS inequality, which uses five observables $A_i$ taken $\pm 1$, shows that there is no hidden variables models in the smallest dimension $d=3$ \cite{klyachko2008simple},
\begin{equation} \label{eq:classicKCBS}
\begin{aligned}
\avg{ \chi_{KCBS}} =\avg{A_1 A_2} + \avg{A_3 A_2} + \avg{A_3 A_4} + \avg{A_5 A_4} + \avg{A_5 A_1} \geq -3.
\end{aligned}
\end{equation}
If all the five observables are predetermined, the inequality of Eq.~\eqref{eq:classicKCBS} always holds. In quantum mechanics, on the other hand, the inequality can be violated for a specific state with properly arranged observables $A_i$. In the case of $d=3$, denote the basis states by $\ket{1}$, $\ket{2}$ and $\ket{3}$. Design the observable $A_{i} = 1-2 \ket{v_{i}}\bra{v_{i}}$ to be the projector along the axis of $\ket{v_{i}}$. The maximal violation of the inequality (\ref{eq:classicKCBS}) is achieved when five state vectors, $\{\ket{v_{i}}\}$, form a regular pentagram, and the initial state vector passes through the center of the pentagram, as shown in Fig.~\ref{fig1:Pentagram}. In this case, the value of $\avg{ \chi_{KCBS}}$ achieves $5- 4 \sqrt{5} \approx -3.944$. The assumption behind the above contextuality inequality is that the observables $A_i$ and $A_{i+1}$ (let $A_6 \equiv A_1$) are compatible. However, in an actual experiment using sequential measurements, the compatibility is difficult to verify, which leads to open the compatibility loophole. The issues of the compatibility in sequential measurements have been addressed by modifying the KCBS inequality \cite{Gunhe10,Szangolies13} (see also Supplementary Materials (SM) I).

We combine the two modifications of the KCBS inequality to relax the condition of the perfect compatibility, which introduce additional terms of $\epsilon$'s \cite{Gunhe10} and $\avg{A_1 A_1}$ \cite{Szangolies13},
\begin{equation} \label{eq:finalKCBS}
\begin{aligned}
\avg{ \chi_{KCBS}} &= \avg{A_1 A_2} + \avg{A_3 A_2} + \avg{A_3 A_4} + \avg{A_5 A_4} + \avg{A_5 A_1} - \avg{A_1 A_1} \\
&\geq -4-(\epsilon_{12}+\epsilon_{32}+\epsilon_{34}+\epsilon_{54}+\epsilon_{51}+\epsilon_{11}).
\end{aligned}
\end{equation}
Here, $\langle A_i A_j \rangle$ denotes the expectation value of the measurement results in the time order of $A_i A_j$ for the sequential measurements. The terms of $\epsilon_{ij}$ describe the difference between a same pair of observables $A_i$ and $A_j$ in different time orders, $A_i A_j$ and $A_j A_i$, which can be regarded as the bound of incompatibility between these sequential measurements \cite{Gunhe10},
\begin{equation} \label{eq:epsilon}
\begin{aligned}
\epsilon_{ij}=\left|\avg{A_j|A_j A_i}-\avg{A_j|A_i A_j}\right|.
\end{aligned}
\end{equation}
The term of $\avg{A_1 A_1}$ is later introduced to address different types of incompatibility, which cannot be excluded with the terms of $\epsilon_{ij}$ \cite{Szangolies13}. In our work, we include both of the modifications that address all types of incompatibility discussed in the Refs \cite{Gunhe10,Szangolies13}. 

\begin{figure}[ht]
\includegraphics[width=0.9\columnwidth]{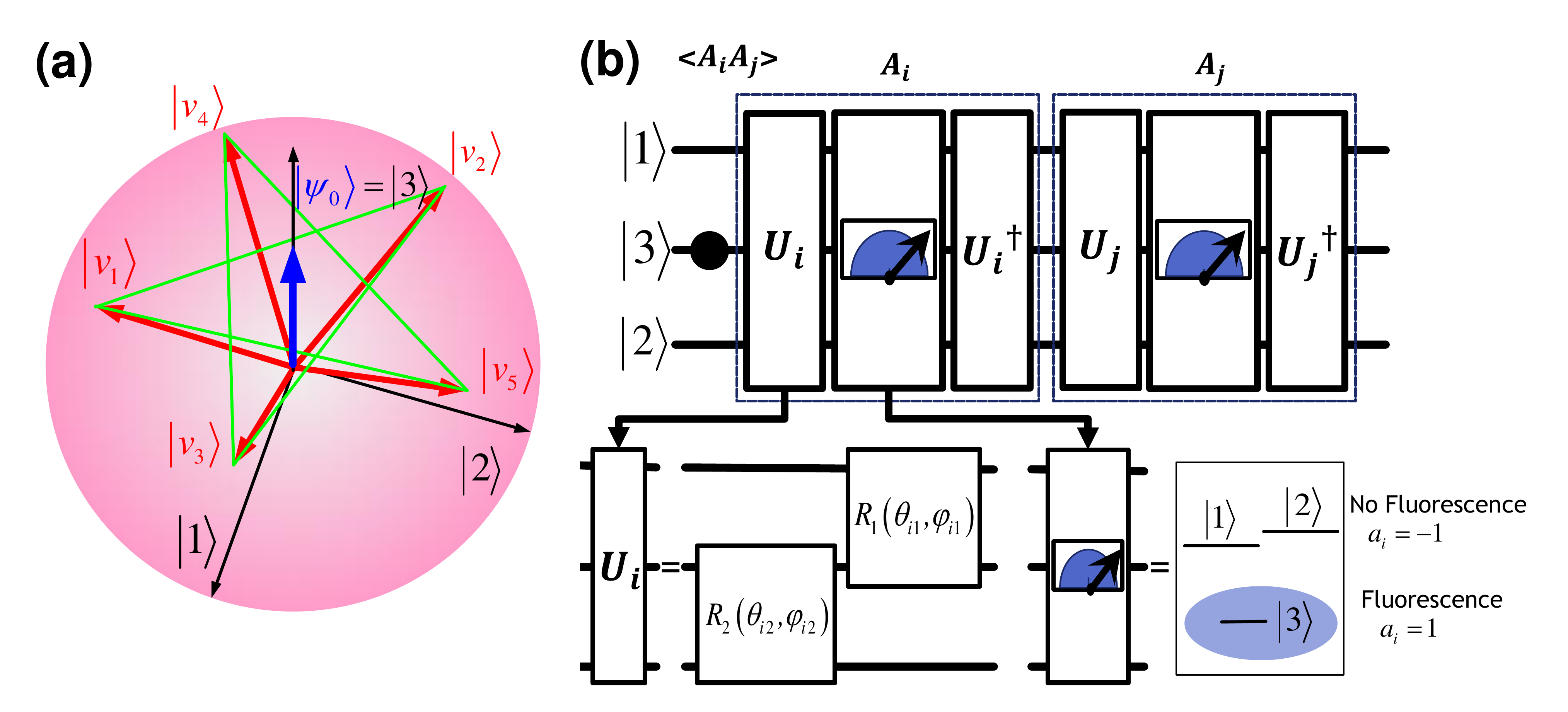}
\caption{KCBS pentagram and experimental procedure. (a) Initial state and five axes which form a pentagram in $d$=3 space. The five observables $A_1, A_2, \dots, A_5$ are the projectors on the axes respectively. The connected axes $\ket{v_i}$ and $\ket{v_{i+1}}$ are orthogonal, representing compatibility of the corresponding observables $A_i$ and $A_{i+1}$. (b) Initially, we prepare $\ket{3}$ state, then perform two sequential measurements of $A_i$ and $A_j$. Each sequential measurement contains a unitary rotation $U_i$, projective measurement, and an inverse unitary rotation $U_{i}^{\dagger}$. Each unitary rotation $U_i$ is comprised of first $R_2 \left(\theta_{2i}, \phi_{2i} \right)$ then $R_1 \left(\theta_{1i}, \phi_{1i} \right)$. In projective measurement, we assign $a_i=1(-1)$ if flourescence is (not) detected.
  }\label{fig1:Pentagram}
\end{figure}

\subsection*{Randomness expansion protocol}
The violation of the KCBS inequality implies the existence of quantum randomness which cannot be imitated by classical variables, which is not only fundamentally interesting but also posses the values for practical applications. The noncontextuality inequalities provide an alternative way of generating secure randomness. Similar to Bell inequality, in each trial, certain bits of randomness are consumed. Thus in order to efficiently expand the randomness from small input randomness, the idea of spot checking is necessary in our scheme. Recently, a robust (error-tolerant) randomness expansion scheme has been proposed \cite{Carl17}, which is a spot-checking protocol that achieves exponential expansion. The protocol is shown in Box \ref{box:scheme}, with our experimental settings.
\begin{figure}[ht]
\fbox{\parbox{16cm}{
\begin{flushleft}

\emph{\bf Denotation}
\begin{itemize}
\item
$G$ : KCBS game with 11 random inputs $\{\{1,2\}, \{2,1\}, \{2,3\}, \allowbreak \{3,2\}, \{3,4\}, \allowbreak \{4,3\}, \{4,5\}, \allowbreak \{5,4\}, \allowbreak  \{5,1\}, \allowbreak \{1,5\}, \allowbreak \{1,1\}\}$ for the game rounds, and the input $\{1,2\}$ is also for the generation rounds

\item
$D$: a quantum device compatible with $G$

\item
Output length $N$: $N_{exp}=1.29 \times 10^8$ in experiment

\item
Test probability $q\in(0,1)$: $q_{exp}=10^{-4}$ in experiment

\item
Score threshold $\chi_g\in(0,1)$: $\chi_g=2/3$ in this KCBS game
\end{itemize}

\emph{\bf Protocol} $R_{gen}$
\begin{enumerate}
\item
Choose a bit $t \in \{0,1\}$ according to the Binomial distribution ($1-q$, $q$).

\item
If $t = 1$ (``game round"), the game $G$ is played with $D$ and the output is recorded. Outputs of game rounds are additionally collected for checking.

\item
If $t = 0$ (``generation round"), $\{1,2\}$ is given to $D$ and the output is recorded.

\item
Steps 1-3 are repeated $N$ times.

\item
Calculate the score $g_{KCBS}$ from all game round outputs. If $g_{KCBS} < \chi_g$, then abort. Otherwise, move to to randomness extraction.
\end{enumerate}
\end{flushleft}
}\label{box:scheme}
}
\caption{The main spot-checking protocol and related denotation.}
\end{figure}

According to the definition of Ref. \cite{Carl17}, the score of the KCBS game is given by $g \in \left\{0,1\right\}$. Thus, Eq.~\eqref{eq:finalKCBS} can be rewritten in the form KCBS game $G$,
\begin{equation}
\begin{aligned}
g_{KCBS} = -\frac{1}{6} (\avg{A_1 A_2} + \avg{A_3 A_2} + \avg{A_3 A_4} + \avg{A_5 A_4} + \avg{A_5 A_1} - \avg{A_1 A_1} \\+\epsilon_{12}+\epsilon_{32}+\epsilon_{34}+\epsilon_{54}+\epsilon_{51}+\epsilon_{11}).
\label{eq:KCBSgame}
\end{aligned}
\end{equation}
The classical winning probability is $\chi_g=2/3$ (see SM.II. for details) and the achievable maximal quantum winning probability is $\chi'_g=(4 \sqrt{5}-4)/6 \approx 0.824$. The gap between $\chi_g$ and $\chi'_g$ enables randomness expansion.

In our scheme, the amount of randomness quantified by the min-entropy is related to the violation of the KCBS inequality (see Methods, Randomness generation rate). For a given game, if the device obtains a superclassical average score, then it must  exhibit certain quantumness, which implies random behavior. This quantum randomness produced by the devices could be extracted. The violation is only based on the observation of experimental data, and can be independent of the sources of prepared states and other device specifications. Therefore, our protocol is self-testing provided that the following assumptions. In our scheme, there are three underlying main assumptions: (1) the input is chosen from an independent random distribution uncorrelated with the system; (2)
the measurement outcomes cannot be leaked directly to adversaries; (3) The first and the second measurements in a context are approximately compatible and can be characterized by $\epsilon_{ij}$ and $\avg{A_1 A_1}$ in Eq. (\ref{eq:finalKCBS}). The assumptions (1) and (2) are widely used in other self-testing tasks, such as device-independent quantum random number generators \cite{Vazirani12,Carl17,Arnon2016}. The assumption (3) is related to the validity of the quantum contextuality test, which would be similar to all the other experimental tests with sequential measurements. We note that we do not require the perfect compatibility. Instead, we assume approximate compatibility, which can be quantified by the terms of $\epsilon_{ij}$ and $\avg{A_1 A_1}$ in Eq. (\ref{eq:finalKCBS}). 
Due to those terms, the violation of the inequality of Eq. (\ref{eq:finalKCBS}) is getting difficult if two sequential measurements are deviated from the perfect compatibility. However, in our scheme, two measurements in a context are performed on a single system, which makes it impossible to exclude the possibility that a malicious manufacturer sabotage the compatibility assumption by registering the setting and results of the first measurements and using them for the second measurements. Therefore, our protocol can not be viewed as a fully-device independent scenario. We need the trust of the device that the measurement settings are close enough to be compatible, but it is fine to have imperfections in the realization and disturbance from classical or quantum noisy environments since the amount of introduced incompatibilities are quantified. Our protocol is well fitted to a scenario of trusted but error-susceptible devices. Given these assumptions, the generated randomness is certified by only experimental statistics.


\subsection*{$^{138}\mathrm{\textbf{Ba}}^+$ qutrit and experimental procedure}
There have been demonstrated the randomness expansion based on the experimental violations of the KCBS inequality using a single trapped \Yb ion \cite{UMark13}. In the demonstration, however, it is not possible to test the modified KCBS inequality, Eq.~\eqref{eq:finalKCBS}, due to lack of capability in obtaining all correlations. For example, when we observe fluorescence in the first measurement, the second measurement does not provide any useful information \cite{UMark13}. Instead, we develop a single \Ba ion system \cite{Dietrich2010,Slodicka2012} with which we can obtain full-correlation results from the sequential measurements by using long-lived shelving states in $^{5}D_{5/2}$ manifold similar to \Ca ion \cite{Leupold2018}. We choose two Zeeman sub-levels ($\ket{m_j =+1/2}\equiv \ket{1}$, $\ket{m_j =+3/2}\equiv \ket{2}$) in the $^{5}D_{5/2}$ manifold, and one Zeeman sub-level ($\ket{m_j =+1/2}\equiv \ket{3}$) in the $^{6}S_{1/2}$ manifold to represent the qutrit system as shown Fig.~\ref{fig3:ExperimentalSetup}(a). In the projective measurement, we observe fluorescence when the state is projected to $\ket{3}$ and no fluorescence for all the other projections on the subspace that consists of $\ket{1}$ and $\ket{2}$ basis while conserving coherence. Different from the \Yb ion realization, since the coherence is not destroyed even when we observe fluorescence in the first measurement, we can get meaningful outcomes in the second measurement. The transitions between $^{6}S_{1/2}$ and $^{5}D_{5/2}$ are coherently manipulated by a narrow-line laser with the wavelength of 1762 nm, which is stabilized to a high-finesse optical cavity. The coherent rotations $R_1 \left(\theta_1, \phi_1 \right)$ between $\ket{1}$ to $\ket{3}$ and $R_2 \left(\theta_2, \phi_2 \right)$ between $\ket{2}$ to $\ket{3}$ (See Methods for the details) are realized by applying the 1762 nm laser beam, where $\theta$ and $\phi$ are controlled by the duration and the phase of the laser beam, respectively, using an AOM.

The procedure of the experimental test of the KCBS inequality consists of Doppler and electromagnetically induced transparency (EIT) cooling \cite{Morigi00,Lin13Sympathetic,Lechner16}, initialization, the first projective measurement of observable $A_i$ and the second projective measurement of $A_j$. The initialization to the state $\ket{3}$ is performed by applying the optical pumping beam of 493 nm with $\sigma^{+}$ polarization shown in Fig. \ref{fig3:ExperimentalSetup}(b). The first measurement of the observable $A_i$ is realized by the rotation $U_i$, the projective measurement, and the reverse of the rotation $U_i^{\dagger}$ (see Methods). The $U_i$ maps the axis $\ket{v_i}$ to the axis $\ket{3}$ and the projective measurement can be described as the projector $M_{\ket{3}}=2\ket{3}\bra{3}-1$ (see Methods). Thus $A_i$ is assigned to value $a_i =1$ when fluorescence is observed and $a_i =-1$ when no fluorescence is observed. The projective measurement consists of the state-dependent fluorescence detection and the optical pumping sequence (see Methods). The second measurement of the observable $A_j$ is realized by the same scheme to that of the first measurement. Unitary rotations of $A_i$(Alice) and $A_j$(Bob) are realized by different signal generators and amplifiers, their results are also collected independently.

\begin{figure}[ht]
\includegraphics[width=0.9\columnwidth]{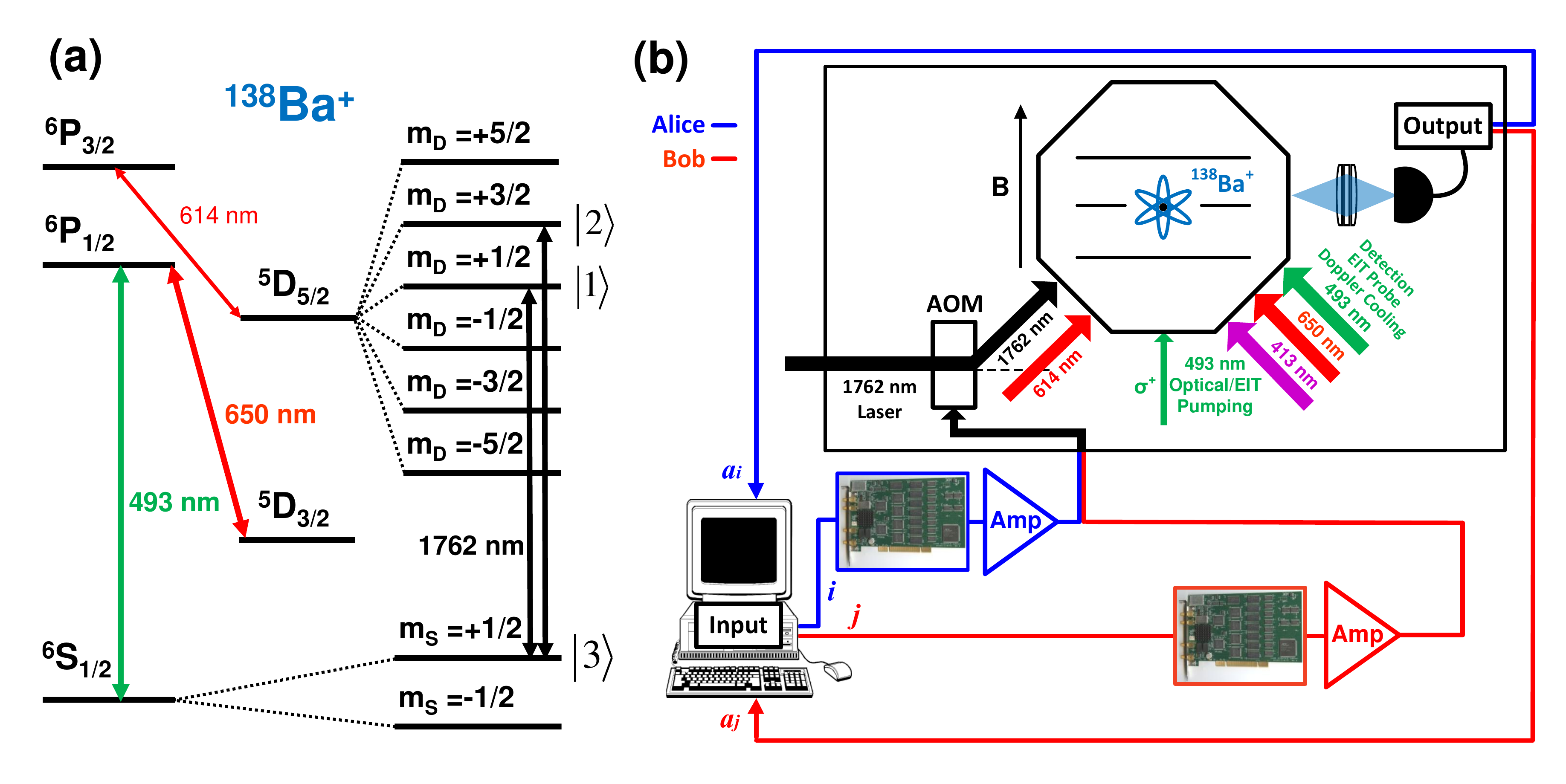}\\
\caption{Experimental setup of the \Ba ion system. (a) The energy level diagram of a \Ba ion for a qutrit system, which is represented by two Zeeman sublevels $\ket{m_D =+1/2}\equiv\mathrm{\ket{1}}$, $\ket{m_D =+3/2}\equiv\mathrm{\ket{2}}$ in the $^{5}D_{5/2}$ manifold, and $\ket{m_S =+1/2}\equiv\mathrm{\ket{3}}$ sublevel in the $^{6}S_{1/2}$ manifold. The quadrupole transitions between $^{6}S_{1/2}$ and $^{5}D_{5/2}$ are coherently manipulated using narrow-line 1762 nm laser which is stabilized to a high-finesse cavity. The 493 nm and 650 nm lasers are used for Doppler cooling, EIT cooling, optical pumping and detection. The 614 nm laser is used for depopulation of $^{5}D_{5/2}$ level to $^{6}S_{1/2}$ level. (b) The experimental setup of a trapped \Ba ion for testing KCBS inequality and for the spot checking random number expansion. One of 11 measurement configurations $\left\{A_i,A_j\right\}$ is randomly selected. When Alice and Bob receive $i$ and $j$, they could not know the setting of the other since each observable is included in at least two different contexts. For example, when Alice receives $i=3$, Bob could be either $j=2$ or $j=4$. Their pulse sequences are independently generated by their own Direct Digital Synthesizer (DDS) and amplifiers, sent to the acousto-optic modulator (AOM) through independent paths, and finally applied to the ion on different time order. Fluorescence is observed by PMT on different time order and the values of the observables are assigned accordingly.
}\label{fig3:ExperimentalSetup}
\end{figure}

\subsection*{Violation of KCBS inequality and randomness expansion}
To test the modified KCBS inequality \eqref{eq:finalKCBS}, we need to measure the eleven combinations of sequential measurements, which include five terms explicitly shown in the inequality \eqref{eq:finalKCBS} as $\avg{A_1 A_2}$, $\avg{A_3 A_2}$, $\avg{A_3 A_4}$, $\avg{A_5 A_4}$, and $\avg{A_5 A_1}$, the other five terms with reverse order ($\avg{A_2 A_1}$, $\avg{A_2 A_3}$, $\avg{A_4 A_3}$, $\avg{A_4 A_5}$, $\avg{A_1 A_5}$), and $\avg{A_1 A_1}$. The reversed-order terms are necessary to observe $\epsilon_{12}$, $\epsilon_{32}$, $\epsilon_{34}$, $\epsilon_{54}$, and $\epsilon_{51}$ and evaluate incompatibility from experimental imperfections. The detailed experimental results of the measurements are summarized in Table \ref{TAB2:KCBSData}.

\begin{table}[htbp]
\caption{
Experimental results for different observables and compatibility terms for the KCBS inequality \eqref{eq:finalKCBS}. Total game rounds are $1.2 \times 10^4$. The standard deviations of the final result are 0.015 and 0.023 for the single observables and correlations, respectively, $10^{-3}$ order for the compatibility terms, all as shown in the parenthesis. The standard deviation for the violation $\sigma$ is 0.068 and our experimental data shows the violation of the extended inequality \eqref{eq:finalKCBS} with 11 $\sigma$.
}
\begin{centering}
\centering
\begin{tabular}[t]
{c|c|c|c|c}
\hline
$\left\{i,j\right\}$ & $\avg{A_i A_j}$ & $\avg{A_i}$ & $\avg{A_j}$ & $\epsilon_{ij}$ \\
\hline
$\left\{\textbf{1,2}\right\}$ & \textbf{-0.768(23)} & 0.082(15) & 0.091(15) & \textbf{0.005(2)} \\
\hline
$\left\{2,1\right\}$ & -0.783(23) & 0.096(15) & 0.065(15) & 0.017(4) \\
\hline
$\left\{2,3\right\}$ & -0.767(22) & 0.098(14) & 0.088(14) & 0.033(5) \\
\hline
$\left\{\textbf{3,2}\right\}$ & \textbf{-0.750(23)} & 0.107(15) & 0.098(15) & \textbf{0.009(3)} \\
\hline
$\left\{\textbf{3,4}\right\}$ & \textbf{-0.773(23)} & 0.084(15) & 0.082(15) & \textbf{0.019(4)} \\
\hline
$\left\{4,3\right\}$ & -0.762(22) & 0.122(14) & 0.068(14) & 0.000(0) \\
\hline
$\left\{4,5\right\}$ & -0.782(23) & 0.095(15) & 0.075(15) & 0.014(3) \\
\hline
$\left\{\textbf{5,4}\right\}$ & \textbf{-0.789(22)} & 0.056(15) & 0.094(15) & \textbf{0.025(4)} \\
\hline
$\left\{\textbf{5,1}\right\}$ & \textbf{-0.773(22)} & 0.100(14) & 0.069(14) & \textbf{0.000(0)} \\
\hline
$\left\{1,5\right\}$ & -0.767(23) & 0.109(15) & 0.066(15) & 0.007(2) \\
\hline
$\left\{\textbf{1,1}\right\}$ & \textbf{0.977(21)} & 0.106(15) & 0.108(15) & \textbf{0.001(1)} \\
\hline
\multicolumn{5}{c}{ $g_{KCBS} = 4.772(68)/6 = 0.795(11)$ }\\
\hline
\end{tabular}
\par\end{centering}\label{TAB2:KCBSData}
\end{table}

For the spot-checking protocol, we choose $\{A_1, A_2\}$ as the setting for generation rounds, i.e., $\left\{1,2\right\}$ as the distinguished input of our KCBS game $G$. At each round, a string of trusted random bits $t$ decides each round is generation round or game round. If it is generation round, we perform the sequential measurement $\{A_1, A_2\}$ and record the output $\{a_1, a_2\}$. If it is game round, we randomly choose one of the 11 configurations of Eq.~\eqref{eq:finalKCBS} and save the result $\{a_i, a_j\}$ after performing the sequential measurement $\{A_i, A_j\}$.

From the Eq.~\eqref{eq:piandO}, we can see that when the violation is small, the total rounds $N$ is a critical parameter. A positive generation rate requires a sufficiently large $N$. Thus we give the minimum required rounds for different violations, which is instructive for experiments.
Figure \ref{fig4:ViolationvsEntropy}(a) shows the minimum total rounds $N_{min}$ to obtain net randomness depending on the KCBS game score $g_{KCBS}$, where $N_{min}$ can be obtained with an optimal $q$. In order to gain net randomness at our experimentally observed $g_{KCBS}=0.795$, we perform $N_{exp}=1.29\times 10^8$ rounds, which is sufficiently larger than $N_{min}=4.6\times 10^7$. At our experimental condition of $N_{exp}$, Fig. \ref{fig4:ViolationvsEntropy}(b) shows the generation rate of net randomness depending on $g_{KCBS}$. If $g_{KCBS}\le0.77$, we can not observe net randomness regardless of $q$. When $g_{KCBS}>0.77$, there exist optimal $q$ values. If $q$ is bigger than proper range, input randomness increases thus no net randomness is produced. If $q$ is smaller than proper range, due to the increase of $\Delta$ in Eq.~\eqref{eq:Rgen}, we also cannot gain net randomness. In our experiment, we choose $q_{exp}=10^{-4}$ as shown in red circle of Fig.~\ref{fig4:ViolationvsEntropy}(b).

\begin{figure}[ht]
\includegraphics[width=0.9\columnwidth]{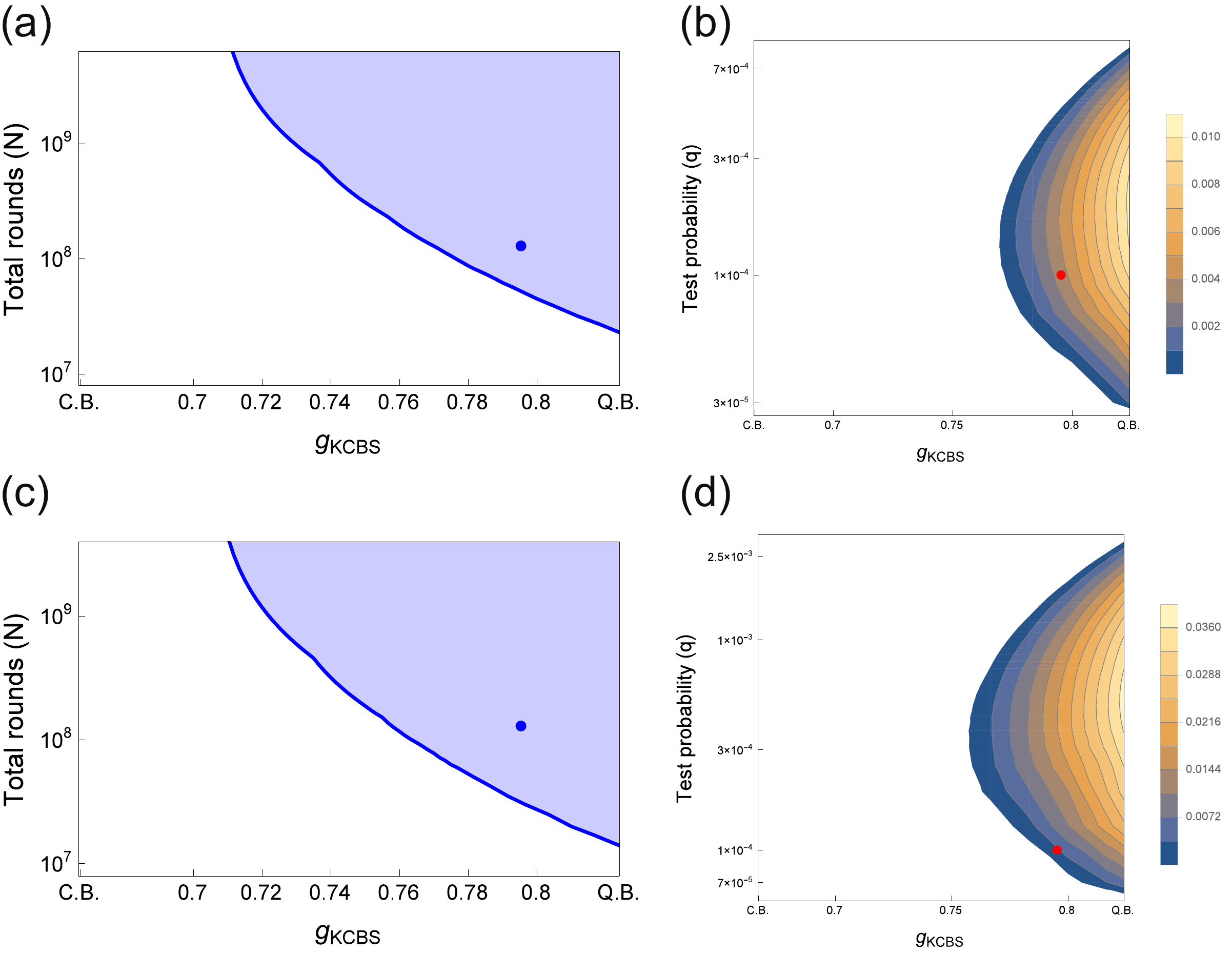}
\caption{(a-b) For MS-bound the relation of the score of KCBS game $g_{KCBS}$, number of total rounds $N$, test probability $q$, and randomness expansion rate with smoothing parameter $\delta=10^{-2}$ in Eq.~\eqref{eq:smoothEntropy}. (a) The minimum number of rounds to have net randomness depending on the score $g_{KCBS}$. The minimum $N$ decreases as $g_{KCBS}$ increases. We can get net randomness only within the shadow area. Our experimental $g_{KCBS}=0.795$ and $N_{exp}=1.29 \times 10^8$ are shown as the green circle. (b) Randomness expansion rate at different $g_{KCBS}$ and $q$ for our $N_{exp}$. Only with the combination of large enough $g_{KCBS}$ and proper $q$ can we obtain net randomness. Our experimental $g_{KCBS}=0.795$ and $q_{exp}=0.0001$ are shown as the red circle, resulting expansion rate $3.4 \times 10^{-3}$ per bit. (c-d) For HS-bound the relation of the score of KCBS game $g_{KCBS}$, number of total rounds $N$, test probability $q$, and randomness expansion rate with smoothing parameter $\delta=10^{-4}$ in Eq.~\eqref{eq:smoothEntropy}. (c) The minimum number of rounds to have net randomness depending on the score $g_{KCBS}$. Our experimental condition is shown as the green circle. (d) Randomness expansion rate at different $g_{KCBS}$ and $q$ for our $N_{exp}$. Our experimental $g_{KCBS}=0.795$ and $q_{exp}=0.0001$ are shown as the red circle, resulting expansion rate $4.4 \times 10^{-3}$ per bit, although our $q_{exp}$ is not optimal for this case.
  }\label{fig4:ViolationvsEntropy}
\end{figure}

Meanwhile, we also apply HS bound to our experimental data as shown in Fig. \ref{fig4:ViolationvsEntropy}. The HS bound produces a bigger generation rate than the MS bound, thus we are able to reduce smoothing parameter $\delta$ to $10^{-4}$, which is the security failure probability. We find that the optimal $q$ for the HS bound is different from that of the MS bound, but our $q_{exp}$ is still good enough to generate net randomness as shown in Fig. \ref{fig4:ViolationvsEntropy}(d).

We play $N_{exp}=1.29 \times 10^8$ (129421072) rounds and observe the left hand side of the inequality Eq.~\eqref{eq:finalKCBS}, $\avg{\chi_{KCBS}} = -4.831$, and the right hand side $-4-(\epsilon_{12}+\epsilon_{32}+\epsilon_{34}+\epsilon_{54}+\epsilon_{51}+\epsilon_{11})=-4.058$.  The detailed experimental results of are summarized in Tab.~\ref{TAB2:KCBSData}. The obtained final score of KCBS game is $g_{KCBS} = 4.772(68)/6 = 0.795(11)$, which violates the inequality \eqref{eq:finalKCBS} by 11 standard deviations. Our test probability is $q_{exp}=10^{-4} \sim O((\mathrm{log}^3 N_{exp})/N_{exp})$, and the required amount of initial random seed is $O(\mathrm{log}^4 N_{exp})$ bits (see SM.III. and IV. for details). The min-entropy of final randomness is $5.3\times 10^{-3}$ per bit, thus the output random bits is $\Theta(N_{exp})$, achieving exponential randomness expansion. In real number, we get $6.88 \times 10^5$ bits of min-entropy which exceeds $2.35 \times 10^5$ bits of input randomness, resulting $4.52 \times 10^5$ net random bits, expansion rate per round is $3.5 \times 10^{-3}$.

When we apply the HS bound to the experimental data, we get larger min-entropy and expansion rate. Note that $\delta$ is two order smaller than that of the MS bound. The min-entropy of final randomness is $6.2 \times 10^{-3}$ per bit, and the expansion rate per round is $4.4 \times 10^{-3}$. We get $8.06 \times 10^5$ bits of min-entropy which exceeds $2.35 \times 10^5$ bits of input randomness, resulting $5.71 \times 10^5$ net random bits. If we use an optimized $q$ based on the calculation using the MS bound, we can get even larger min-entropy and expansion rate.

\section*{Discussion and Outlook}
In this work, we achieve an exponential randomness expansion secured by quantum contextuality. Regardless of imperfections and experimental noises, the observed violation of the modified KCBS inequality, Eq.~\eqref{eq:finalKCBS}, verifies the generated randomness. In our protocol, we can guarantee the randomness without the i.i.d.~assumption even when imperfections or noises may originate from quantum mechanics, which would be our quantum adversary. Note that there are other types of quantum contextuality inequalities that do not require sequential measurements, which could also ensure the no-disturbance condition. Our work can be easily extended to these proposals as well.

Due to the advantage of using contextuality for randomness certification, our current generation speed is 270 bits s$^{-1}$ and 1.7 bits s$^{-1}$ after applying Toeplitz matrix hashing, which is faster than that of using Bell's inequality \cite{Pironio10,bierhorst2018experimentally}. We believe we can achieve orders of magnitude higher generation speed by several improvements in duration of cooling, optical pumping, and detection, coherence time of qutrit, and coherent operation time (see SM.V. for details). From the theoretical aspect, though the generation rate used in our scheme is robust and noise-tolerable, a large number of trials are still required which costs a lot of efforts. An improved generation rate based on general contextuality inequality is still an open problem. Recently, entropy accumulation theory has been applied in device-independent protocols \cite{dupuis2016entropy,arnon2018practical} and may be a potential tool for achieving a near optimal generation rate using contextuality inequality.

Fully device-independent random number generation puts a very high requirement on implementation devices. In practice, it is meaningful to pursue alternative randomness generation schemes with additional reasonable assumptions, such as Bell test with certain loopholes \cite{Liu2018High}, uncertainty principles, or contextuality \cite{kulikov2017realization}. Our scheme is not fully device-independent due to the approximate compatibility assumption on measurements. On the other hand, our scheme does enjoy the self-testing properties on both source and measurement. Note that the self-testing protocols with proper assumptions on the device have also been proposed to deal with other quantum information processing tasks \cite{lunghi2015self,fiorentino2007secure}.

The security proof in \cite{Carl17} only considers the perfect case without imperfections of compatible or no-disturbance. Here we characterize this imperfections and modify the score of KCBS game. 
We assume the imperfections in experiments does not affect the adversary and security proof in \cite{Carl17} and only leads to a modified classical bound. The rigorous proof of a self-testing random number generator with limited compatibility is an interesting open problem and we will leave it as a future theoretical work.

Moreover, quantum contextuality can also provide an alternative means for randomness amplification. In principle, we can individually manipulate multiple ions and use them to generate random numbers simultaneously, which could lead to orders of magnitude faster generation speed. Such kind of multiple ion system can be applied to realize randomness amplification protocol \cite{Chung2014Physical}, which generates true randomness out of weak randomness input. The protocol can be implemented by the multiple of our developed randomness expansion systems and the exclusive-OR of their outputs.

\subsection*{Data availability}
The authors declare that the main data supporting the finding of this study are available within the article and its Supplementary Material files. Additional data can be provided by the corresponding author upon request.

\section*{Methods}
\subsection*{Randomness generation rate}
Here, we consider the case that the average probability
of measurement setting choice is unbiased, $p(a) = 1/11$, $a\in \{(i,i+1),(i+1,i),(1,1)\} (i=1,2,\dots,5)$.
The violation of the inequality in Eq.~\eqref{eq:finalKCBS}, indicates the presence of genuine quantum randomness in the measurement outcomes. The amount of secure randomness can be quantified by the smooth min-entropy $H_{min}^\delta(X|AE)$, which is bounded by
\begin{equation} \label{eq:smoothEntropy}
\begin{aligned}
H_{min}^\delta(X|AE)\ge N R_{gen} (g_{KCBS}, q,\epsilon,N, \delta),
\end{aligned}
\end{equation}
where $X$ and $A$ denote the output and input sequences, respectively; $E$ denotes the system of an quantum adversary; $\delta$ is the smoothing parameter representing the security failure probability; $g_{KCBS}$ is the KCBS game score; $N$ is the total number of experiment trials; $q$ is the probability of choosing game round; $\epsilon$ is the parameter of Schatten norm, in the security analysis, $(1+\epsilon)$-Schatten norm is applied; $R_{gen}$ is the lower bound of randomness generation on average for each trial. In order to achieve the maximal randomness expansion, we also need to consider the input randomness for each trial,
\begin{equation} \label{eq:inputrandomness}
\begin{aligned}
R_{In}=q\log11+H(q),
\end{aligned}
\end{equation}
and the randomness expansion rate can be expressed as $R_{exp}=R_{gen}-R_{In}$. The output randomness rate $R_{gen}$ is given by
\begin{equation}
\begin{aligned}
R_{gen}=\pi(\chi)-\Delta,
\label{eq:Rgen}
\end{aligned}
\end{equation}
where
\begin{equation} \label{eq:piandO}
\begin{aligned}
&\chi=g_{KCBS}-\chi_g, \\
&\pi(\chi)=2\frac{\log(e)\chi^2}{r-1}, \\
&\Delta=\frac{\epsilon}{q} \frac{8\log(e)\chi^2}{(r-1)^2} +
\frac{\log(2/\delta^2)}{N\epsilon} + 2rq + O\left(\left(\frac{\epsilon}{q}\right)^2\right).
\end{aligned}
\end{equation}
Here, all the log is base 2 throughout the paper, $r$ is the output alphabet size, which is $r=4$ in our KCBS game. The explicit form of $O\left(\left(\frac{\epsilon}{q}\right)^2\right)$ and derivation of Eq.~\eqref{eq:piandO} are shown in Sections III and IV of Supplementary Materials. Denote the above bound as Miller-Shi (MS) bound \cite{Carl17} and afterwards a tighter bound is obtained, referred as Huang-Shi (HS) bound without the dependence of $r$ \cite{Huang17}. For the experiment, we perform the parameter optimization of $q$ and $\epsilon$ to achieve the maximal randomness expansion rate $R_{exp}$ with MS bound and also show the final randomness rate for two different bounds.

\subsection*{Unitary Rotations}
Here, $R_1 \left(\theta_1, \phi_1 \right)$ and $R_2 \left(\theta_2, \phi_2 \right)$ are defined as
\begin{equation*}
R_1 \left(\theta_1, \phi_1 \right) = \left(%
\begin{array}{ccc}
\text{$\cos$}\frac{\theta_1}{2} & 0 & -ie^{i\left(\phi_1+\frac{\text{$\pi$}}{%
2}\right)}\text{$\sin$}\frac{\theta_1}{2} \\
0 & 1 & 0 \\
-ie^{-i\left(\phi_1+\frac{\text{$\pi$}}{2}\right)}\text{$\sin$}\frac{\theta_1%
}{2} & 0 & \text{$\cos$}\frac{\theta_1}{2}%
\end{array}%
\right),
\end{equation*}

\begin{equation*}
R_2 \left(\theta_2, \phi_2 \right) = \left(%
\begin{array}{ccc}
1 & 0 & 0 \\
0 & \text{$\cos$}\frac{\theta_2}{2} & -ie^{-i\left(\phi_2+\frac{\text{$\pi$}%
}{2}\right)}\text{$\sin$}\frac{\theta_2}{2} \\
0 & -ie^{i\left(\phi_2+\frac{\text{$\pi$}}{2}\right)}\text{$\sin$}\frac{%
\theta_2}{2} & \text{$\cos$}\frac{\theta_2}{2}%
\end{array}%
\right).
\end{equation*}

The Unitary rotations $U_i$ in the measurement configurations shown in Fig. \ref{fig1:Pentagram}(b) are realized by corresponding $R_2 \left(\theta_{2i}, \phi_{2i} \right)$ then $R_1 \left(\theta_{1i}, \phi_{1i} \right)$, while $U_{i}^{\dagger}$ are composed of $R_1 \left(\theta_{1i}, \pi-\phi_{1i} \right)$ then $R_2 \left(\theta_{2i}, \pi-\phi_{2i} \right)$, where the specific $U_i$ are listed in Tab. \ref{TAB1:Rotation}.

\begin{table}[htbp]
\caption{
Unitary rotations $U_i$.
}
\begin{centering}
\centering
\begin{tabular}[t]
{c|c}
U & Rotation \\
\hline
$U_1$ & $R_1(0.531\pi,\pi)\cdot R_2(0.066\pi,0)$ \\
$U_2$ & $R_1(0.442\pi,0)\cdot R_2(0.328\pi,0)$ \\
$U_3$ & $R_1(0.191\pi,\pi)\cdot R_2(0.506\pi,\pi)$ \\
$U_4$ & $R_1(0.104\pi,\pi)\cdot R_2(0.526\pi,0)$ \\
$U_5$ & $R_1(0.377\pi,0)\cdot R_2(0.404\pi,\pi)$ \\

\end{tabular}
\par\end{centering}\label{TAB1:Rotation}
\end{table}

\subsection*{Experimental sequence}
Each round comprises Doppler cooling, EIT cooling, optical pumping, rotation ($U_i$), the first projective measurement, inverse rotation ($U_{i}^{\dagger}$), rotation ($U_j$), the second projective measurement, inverse rotation ($U_{j}^{\dagger}$). The \Ba ion is first cooled down with 500 $\mu$s Doppler cooling and 1000 $\mu$s EIT cooling. Optical pumping procedure initializes the internal state of the ion to $\ket{m_S=+1/2}$ by carefully adjusting the polarization of 493 nm laser beam. We manipulate the states between $\ket{1}$ and $\ket{3}$, and between $\ket{2}$ and $\ket{3}$ by applying 1762 nm laser with different frequencies and amplitudes controlled by AOM. The 1762 nm fiber laser is stabilized with a high-finesse cavity to achieve a linewidth below 1 Hz using Pound-Drever-Hall technique. The cavity is made of ultra-low-expansion material and is mounted in a vacuum cavity with active temperature stabilization to maximize the stability of its length. Frequency and amplitude of RF signal for AOM inputs are generated by two independent pairs of DDS (AD9910) for $A_i$ and $A_j$ measurements, which represent Alice and Bob, ensuring they are compatible without communication. The 2$\pi$ time for both Rabi oscillations are adjusted to 37 $\mu$s, that is $\Omega=\left(2\pi\right)27$ kHz. Every rotation $U_i$ is performed with same duration of no longer than 16 $\mu$s.

EIT cooling implements the asymmetry profile of the absorption spectrum to cancel the heating effect caused by carrier transition meanwhile strength the red-sideband transition to hold the cooling function \cite{Morigi00,Lin13Sympathetic,Lechner16}. EIT cooling only need three level, however there are four Zeeman states of \Ba ion. Though with only doppler cooling and EIT cooling the ion is not perfectly cooled to the ground state without sideband cooling (average phonon number $\avg{\bar{n}=0.1}$), the carrier transition operated by stabilized 1762 nm laser has enough fidelity due to the small Lamb-Dicke parameter $\eta=0.07$.

Our projective measurement includes state discrimination and state re-preparation. We differentiate one state versus the other two states of a qutrit using the standard fluorescent-detection method. For the $\ket{3}$ state, average of 32 photons at 493 nm can be detected during 600 $\mu$s and no photons for the $\ket{1}$ or the $\ket{2}$ state. In experiment, perfect state detection fidelity is achieved for $\ket{3}$, while the error of $\ket{1}$ and $\ket{2}$ is 1.3\%. Duration of the first projective measurement is set to 600 $\mu$s with discrimination $n_{\mathrm{ph}}=3$ while the second projective measurement is 300 $\mu$s and $n_{\mathrm{ph}}=1$. Fluorescence detection duration is longer than the coherence time between $\ket{1}$ and $\ket{2}$, which is around 200 $\mu$s. Therefore we add spin echo pulses during the fluorescence detection to keep the coherence until the second measurement is done. Re-preparation to $\ket{3}$ state, which is realized by optical pumping without 614 nm laser, keeps the coherence between $\ket{1}$ and $\ket{2}$ in $^{5}D_{5/2}$ manifold. Since the second projective measurement is the end of the experiment without further operations, we do not apply spin echo pulses and state re-preparation, which results in shorter duration.

\subsection*{Extractor and random test} A random number extractor is a hashing function transforming a non-perfect random number string $\{0,1\}^N$ to a nearly perfect one $\{0,1\}^m$. In our experiment, the length of the input string is $N_{exp}=1.29\times 10^8$ and $H_{min}(X|IE)=6.2\times 10^{-3}$ per bit. According to leftover hash lemma \cite{impagliazzo1989pseudo}
\begin{equation}
m \leq NH_{min}(X|IE)-2\log{\frac{1}{\epsilon_h}},
\end{equation}
we set the security parameter $\epsilon_h$ to be a typical value $\epsilon_h=2^{-100}$, and the length of the output string is $m=8.06 \times 10^5$.
Here we apply a random $m \times N_{exp}$ Toeplitz matrix \cite{wegman1981new} as the hashing function. The input random seed $\{0,1\}^s$ ($s=m+N_{exp}-1$) is from \cite{nie2015generation}.

We apply the random test \cite{Rukhin10} to the extracted data. The tests include 'Frequency', 'Block Frequency (BFreq)', two 'Cumulative Sums (CuSm)' tests, 'Runs', 'Longest-Run-of-Ones in a Block (LROB)', 'Rank', 'Fast Fourier Transform (FFT)', 'Serial'. The $p$-values are distributed in the interval (0, 1), which show the probabilities that an ideal random number generator would produce less random sequence than the tested one. If $p$-value is taken 0, it means the tested data is fully non-random, while 1 means completely random. The threshold we set for accepting the data as random is 0.01. As shown in Fig. \ref{fig5:RandomTest}, the outputs strings ${a_i}^N$ and ${a_j}^N$ pass all tests. However, as expected, the combined outputs ${(a_i a_j)}^N$ do not pass all tests because since the measurement outputs of two observables are correlated thus are not independent random variables.

\begin{figure}[ht]
\includegraphics[width=0.9\columnwidth]{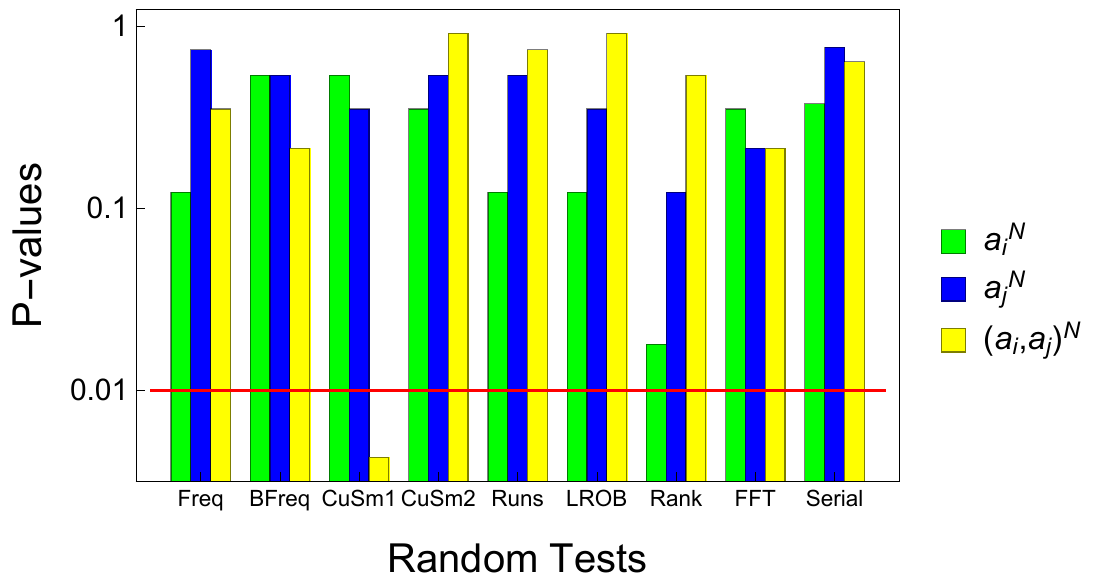}
\caption{The results for random tests \cite{Rukhin10} of the outputs of the first measurement $a_i$ and the second measurement $a_j$, and both measurement $a_i a_j$. Outputs of ${a_i}^N$ and ${a_j}^N$ pass the listed tests since all $p$-values exceed the threshold 0.01, while the outputs of ${(a_i a_j)}^N$ failed to pass the first test of 'Cumulative Sums (CuSm)'.
  }\label{fig5:RandomTest}
\end{figure}

\section*{Acknowledgements}
We thank Yaoyun Shi, Carl Miller, Kai-Min Chung, Cupjin Huang, and Xiao Yuan for helpful discussions. This work was supported by the National Key Research and Development Program of China under Grants No.~2016YFA0301900, No.~2016YFA0301901, No.~2017YFA0303900, and  No.~2017YFA0304004, and the National Natural Science Foundation of China Grants No.~11374178, No.~11574002, No.~11674193, and No.~11875173.

\subsection*{Corresponding authors} Correspondence to Xiongfeng Ma (xma@tsinghua.edu.cn) and Kihwan Kim (kimkihwan@mail.tsinghua.edu.cn).

\clearpage
\onecolumngrid

\begin{center}
\large{\textbf{Supplementary Materials: Randomness expansion secured by quantum contextuality}}
\end{center}

\section{Modified noncontextual inequality}
Among the KS inequalities, KCBS inequality, which uses five observables $A_i$ taking $\pm 1$, shows that with noncontexual hidden variables, the l.h.s of the inequality is no less than -3 \cite{klyachko2008simple},
\begin{equation}
\begin{aligned}
\avg{A_1 A_2} + \avg{A_3 A_2} + \avg{A_3 A_4} + \avg{A_5 A_4} + \avg{A_5 A_1} \geq -3.
\end{aligned}
\end{equation}
In practice, the observables $\avg{A_i A_j}$ have to be implemented in a sequential measurement. We denote the observalble $A_i$ with superscript $m$, $A_i^m$ as the measurement of $A_i$ at the position $m$ in the sequence. For example, $A_i^1A_j^2$ denotes the sequence of measuring $A_i$ first, then $A_j$.

Noncontexual HV model requires that the outcomes of any observable $A_i$ does not depend on other compatible jointly measured observables with $A_i$. To be more specific, we take $A_1$ as an example. It is compatible with $A_2$ and $A_5$. We denote the obtained value as $v$, then have $v(A_1^1)=v(A_1^2|A_2^1 A_1^2)$ and   $v(A_1^1)=v(A_1^2|A_5^1 A_1^2)$.

The assumption behind the above contextuality inequality is that the observables $A_i$ and $A_{i+1}$ (let $A_6 \equiv A_1$) are compatible. However, in an actual experiment using sequential measurements, the compatibility is not perfect which leads to the compatibility loophole.

In \cite{Gunhe10}, this imperfection can be quantified by
\begin{equation}
\begin{aligned}
p^{flip}[A_1A_2]=p[(A_2^1(+)|A_2^1)~and~(A_2^2(-)|A_1^1 A_2^2)]+p[(A_2^1(-)|A_2^1)~and~(A_2^2(+)|A_1^1 A_2^2)].
\end{aligned}
\end{equation}
Here $+,-$ denote the obtained value and this probability can be understood as the $A_1$ flips the predetermined value of $A_2$.
Then using the fact $\avg{A_1 A_2}\le \avg{A_1^1 A_2^2}+2p^{flip}[A_1A_2]$, the inequality can be modified as
\begin{equation}
\begin{aligned}
&\avg{A_1^1 A_2^2} + \avg{A_3^1 A_2^2} + \avg{A_3^1 A_4^2} + \avg{A_5^1 A_4^2} + \avg{A_5^1 A_1^2} \geq \\ &-3-2(p^{flip}[A_1A_2]+p^{flip}[A_3A_2]+p^{flip}[A_3A_4]+p^{flip}[A_5A_4]+p^{flip}[A_5A_1]).
\end{aligned}
\end{equation}
Note that this inequality holds for any HV models. 
In the experiment, $p^{flip}$ is not achieveable and different approaches are proposed to estimated with different assumptions. 
Here we use $\epsilon_{ij}$ to quantify the difference between a same pair of obervables $A_i$ and  $A_j$ in different time order, $A_i A_j$ and $A_j A_i$, which can be regarded as the bound of incompatibility of these sequential measurements,
\begin{equation}
\begin{aligned}
\left|\avg{A_j|A_j A_i}-\avg{A_j|A_i A_j}\right| \le \epsilon_{ij}.
\label{eq:epsilon}
\end{aligned}
\end{equation}
For experimentally accessible distributions, 
\begin{equation} 
\begin{aligned}
\left|p(A_i=a| A_iA_{i+1})- p(A_i=a| A_{i+1}A_i)\right|\le \epsilon_{ij}/2, 
\end{aligned}
\end{equation}
where $a\in \{+,-\}$.
We assume that the underlaying probability distributions have the same properties as all accessible distributions. Then $p^{flip}[A_1A_2]$ can be bounded by $\epsilon_{12}/2$ which is obtained in the experiments, $p^{flip}[A_1A_2]\le \epsilon_{12}/2$.
However, the probability distributions of a general HV model may not belong to the set of experimentally accessible probability distributions. We assume that this difference is negligible and that the properties verified in accessible experiments hold also for some of HV models.

Combining another modification in \cite{Szangolies13}, we apply an extended version of KCBS inequality
\begin{equation}
\begin{aligned}
\avg{ \chi_{KCBS}} = \avg{A_1 A_2} + \avg{A_3 A_2} + \avg{A_3 A_4} + \avg{A_5 A_4} + \avg{A_5 A_1} - \avg{A_1 A_1} \geq \\-4-(\epsilon_{12}+\epsilon_{32}+\epsilon_{34}+\epsilon_{54}+\epsilon_{51}+\epsilon_{11}),
\label{eq:finalKCBS}
\end{aligned}
\end{equation}
here for simplicity, we omit the time order superscript and $\langle A_i A_j \rangle$ denotes the expectation value of the measurement results in the time order of $A_i A_j$ for the sequential measurements; 

The above modifications of the inequality can be understood from the point of view of the game, which is played by two players $Alice$ and $Bob$ who receive random inputs for measurement settings without knowing the other's, similar to the Bell-inequality nonlocal game \cite{Colbeck07,Vazirani12,Colbeck2011private}. The score of each trial is calculated according to the inputs and outputs.
Each nonlocal game can be transformed into a contextuality game because no-communication local measurements is a stronger assumption and satisfy the compatible assumption. But on the contrary, not every contextuality game can be transformed into a nonlocal game.
The inequality with only terms of $\epsilon_{ij}$ is not a Bell inequality because it can also be violated by a simple classical strategy, two players output always opposite results.
Thus it is critical to have the term of $-\avg{A_1 A_1}$. In the following, it can be proved that the modified KCBS inequality even without $\epsilon_{ij}$ terms is a Bell inequality which cannot be violated by all classical local hidden means.
Inspired by a modified KCBS inequality, we propose a new Bell inequality, we assume that the measurements in different time order can not communicate with each other. With local hidden variable, the l.h.s of the inequality is no less than -4.
\begin{equation}
\begin{aligned}
\avg{A_1 A_2} + \avg{A_3 A_2} + \avg{A_3 A_4} + \avg{A_5 A_4} + \avg{A_5 A_1}-  \avg{A_1 A_1}\geq -4.
\end{aligned}
\end{equation}

\begin{proof}
\begin{equation}
\begin{aligned}
&\avg{A_1 A_2} + \avg{A_3 A_2} + \avg{A_3 A_4} + \avg{A_5 A_4} + \avg{A_5 A_1}-  \avg{A_1 A_1}\\
=&\avg{A_1 A_2} + \avg{A_3 A_2} + \avg{A_3 A_4} -\avg{A_1 A_4} + \avg{A_5 A_4} +    \avg{A_1 A_4}     +\avg{A_5 A_1}-  \avg{A_1 A_1}\\
\ge & \avg{A_1(A_2-A_4)}+ \avg{A_3(A_2+A_4)}-2
\end{aligned}
\end{equation}
The inequality holds because with local hidden variable, $\avg{A_5 A_4} + \avg{A_1 A_4}+\avg{A_5 A_1}-\avg{A_1 A_1}\ge -2$, which is a CHSH inequality. $A_i\in \{\pm 1\}$, either $A_2+A_4=0$ or $A_2-A_4=0$ will hold, thus $\avg{A_1(A_2-A_4)}+ \avg{A_3(A_2+A_4)}\ge -2$. The l.h.s is no less than -4 with local hidden variable.
\end{proof}
From the view of nonlocal game, it is critical to have the term $-\avg{A_1 A_1}$ in Eq.~\eqref{eq:finalKCBS}.

\section{Miller and Shi's security proof and its feasibility in practical case}
Here in this section, we mainly focus on the work \cite{Carl17} and overview their security proof.

The min entropy is used for evaluating the randomness. Given the output $X$ , conditioned on input $A$ and adversary' system $E$, the smooth min entropy $H_{min}^\delta(X|AE)$ is defined as
\begin{equation}
\begin{aligned}
H_{min}^\delta(X|AE)= \max \limits_{\lVert \Gamma^{'}-\Gamma_{AEX}\rVert\le \delta} H_{min}(X|AE)_{\Gamma^{'}}
\end{aligned}
\end{equation}

The direct estimation of min entropy is generally hard, thus their security proof applied Renyi entropy to give the lower bound of min entropy. For a quantum state $\rho$, its smooth min-entropies satisfy
\begin{equation}
\begin{aligned}
H_{min}^\delta(\rho)= H_{1+\varepsilon} (\rho)-\frac{\log(1/\delta)}{\varepsilon}
\end{aligned}
\end{equation}
where  $H_{1+\varepsilon} (\rho)=-\frac{1}{\varepsilon}\log \mathrm{Tr}[\rho^{1+\varepsilon}]$.
The randomness in its output is quantified by this $(1+\varepsilon)$-randomness.
The main tool proposed in this proof is a  $(1+\varepsilon)$-uncertain relation.
After a projective measurement, the amount of randomness ($(1+\varepsilon)$-randomness) obtained from a measurement is related to the degree of disturbance caused by the measurement, shown in Proposition 4.4. For a given fixed input, the device has a classically predicable output and can achievable maximal score is $w$. Then if device obtains a score higher than this threshold $w$, then there must be unpredictable randomness in the output of this device. The rate curve is achieved in Corollary 6.11. This security proof is general for not only nonlocal game but also for contextuality. The uncertain relation is only relevant to the size of output alphabet and the measurement in contextuality can fit this proposition. For different schemes, the major differences is the classically predicable bound $w$. Note that this bound $w$ is the maximal score for devices which has classically predictable outputs on an input. It is different with the classical strategy bound by hidden variable $C_G$ in general. Though different in the definition, the value can be the same for some specific cases, for example, nonlocal game with binary input in each party and contextuality shown in Appendix D of \cite{Carl17}. However, in the practical case, the measurements in contextuality is not compatible. Though the uncertain relation in Proposition 4.4 still holds, the remained problem is to calculate $w$ and check whether it equals to the classical bound achieved by approximately contextuail hidden variable. We express this KCBS game as
\begin{equation}
\begin{aligned}
G(A_1,A_2,A_3,A_4,A_5)=-\frac{1}{6}(A_1^1A_2^2+A_3^1A_2^2+A_3^1A_4^2+A_5^1A_4^2+A_5^1A_1^2-A_1^1A_1^2+\epsilon_{12}+\epsilon_{32}+\epsilon_{34}+\epsilon_{54}+\epsilon_{51}+\epsilon_{11}).
\end{aligned}
\end{equation}

\begin{proposition}
Let G be the game given above , $w=2/3$
\end{proposition}
\begin{proof}
With the approximately noncontextual hidden variable, the maximal score is $C_G=2/3$. This strategy is classically predictable, thus the maximal score $w$ with an input classically predictable should not be less than $C_G$, i.e. $w\ge C_G$. We suppose that there is a device $D$ (can be quantum) applied in KCBS game which outputs a score above 2/3, and which gives a deterministic output on input 1,
\begin{equation}
\begin{aligned}
-4\ge \avg{\chi_{KCBS}},
\end{aligned}
\end{equation}
where $\avg{\chi_{KCBS}}= \avg{A_1^1A_2^2}+\avg{A_3^1A_2^2}+\avg{A_3^1A_4^2}+\avg{A_5^1A_4^2}+\avg{A_5^1A_1^2}-\avg{A_1^1A_1^2}+\epsilon_{12}+\epsilon_{32}+\epsilon_{34}+\epsilon_{54}+\epsilon_{51}+\epsilon_{11}$ is the practical mean value with sequential measurements.
Due to $ \avg{A_i A_j}\ge -1+| \avg{A_i}+ \avg{A_j}|$, $\avg{A_i A_j}\le \avg{A_i^1 A_j^2}+2p^{flip}[A_iA_j]$ and $p^{flip}[A_iA_j]\le \epsilon_{ij}$, we have
\begin{equation}
\begin{aligned}
\avg{\chi_{KCBS}}&\ge -6+ | \avg{A_1}+ \avg{A_2}|+| \avg{A_3}+ \avg{A_2}|+| \avg{A_3}+ \avg{A_4}|+| \avg{A_5}+ \avg{A_4}|+| \avg{A_5}+ \avg{A_1}|\\
&\ge -6+ | \avg{A_1}+ \avg{A_2}|+| \avg{-A_2}-\avg{A_3}|+| \avg{A_3}- \avg{-A_4}|+| \avg{-A_4}- \avg{A_5}|+| \avg{A_5}- \avg{-A_1}|.
\end{aligned}
\end{equation}
Therefore, with the triangle inequality,
\begin{equation}
\begin{aligned}
-4\ge -6+|\avg{A_1}-\avg{-A_1}|.
\end{aligned}
\end{equation}
The fixed input 1 is deterministic, thus $\avg{A_1}=\pm 1$, this is a contradiction. Thus $w\le C_G=2/3$ and $w=2/3$.
\end{proof}
With this proposition, any score above $w$ can be used to generate randomness though the observables are approximately compatible.

\section{Randomness generation rate}
Here in this section, based on the work \cite{Carl17} we give an exact result for the randomness expansion rate.
The min entropy is used for evaluating the randomness.
Combining Theorem 4.1 and Proposition 6.8 in \cite{Carl17} yields
\begin{equation}
\begin{aligned}\label{rate1}
&H_{min}^\delta(X|AE)\ge N [\pi(\chi)-O(q+\epsilon/q+\frac{\log(2/\delta^2)}{N\epsilon})]\\
\end{aligned}
\end{equation}
where $O(\frac{\log(2/\delta^2)}{N\epsilon})$ and $O(q+\epsilon/q)$ come from Theorem 3.2 and Proposition 6.8, respectively. From Theorem 3.2, we can let
$O(\frac{\log(2/\delta^2)}{N\epsilon})=\frac{\log(2/\delta^2)}{N\epsilon}$. $O(q+\epsilon/q)$ comes from Proposition 6.5, the combination of Proposition 6.3 and 6.4.
In the proof of Proposition 6.4, from Eq.(6.25) to Eq.(6.26) is equivalent to
\begin{equation}
\begin{aligned}\label{rate2}
\frac{\sum_x \langle\rho^x_{\bar{a}} \rangle_{1+\epsilon}}{\langle\rho \rangle_{1+\epsilon}}\ge 1-O(\epsilon)
\end{aligned}
\end{equation}
where $x$ is the output with output alphabet size $r$, and $\bar{a}$ is the input.
According to the Proposition B.2 and Proposition B.3 in Carl's paper, we apply the induction, $\sum_x \langle\rho^x_{\bar{a}} \rangle_{1+\epsilon}\ge (1-\epsilon)^r \langle\sum_x \rho^x_{\bar{a}} \rangle_{1+\epsilon}$ and $\langle\sum_x \rho^x_{\bar{a}} \rangle_{1+\epsilon}\ge (1-\epsilon)^r \langle\rho \rangle_{1+\epsilon}$. Thus $\frac{\sum_x \langle\rho^x_{\bar{a}} \rangle_{1+\epsilon}}{\langle\rho \rangle_{1+\epsilon}}\ge (1-\epsilon)^{2r}
\ge 1-2r\epsilon$ and $O(\epsilon)=2r\epsilon$. Consequently, the term in Proposition 6.4 $O(q)=2rq$.

The estimation in Proposition 6.3 comes from the second order terms in Taylor expansion in Eq.(6.20) and Eq.(6.21). For a function $F(x)$, its Taylor expansion at $a$ is as follows,
\begin{equation}
\begin{aligned}\label{rate3}
F(b)=F(a)+F^{'}(a)(b-a)+\frac{F^{''}(a)}{2}(b-a)^2+\frac{F^{'''}[a      +\theta(b-a)]}{6}(b-a)^3,\theta\in (0,1)
\end{aligned}
\end{equation}
where the fourth term is third order Taylor Lagrange remainder. Here $F(b)=2^{\epsilon s H(a,x)/q}$ and $a=0$.
\begin{equation}
\begin{aligned}\label{rate4}
&2^{\epsilon s H(a,x)/q}-1=\epsilon s \left(\ln2\right) H(a,x)/q+\frac{1}{2}\left(\frac{\epsilon s \left(\ln2\right) H(a,x)}{q}\right)^2+R_3  \\
&R_3=\frac{1}{6}\left(\frac{\epsilon s \left(\ln2\right) H(a,x)}{q}\right)^3 2^{\theta \epsilon s H(a,x)/q}, \theta\in (0,1)
\end{aligned}
\end{equation}
where the term $R_3$ is the third order Taylor Lagrange remainder.
Substitute this expression in Eq.(6.20), we have
\begin{equation}
\begin{aligned}\label{rate5}
&\sum_{a,x} p(a)\left[\frac{1}{2}\left(\frac{\epsilon s\left(\ln2\right) H(a,x)}{q}\right)^2+R_3\right] \langle\rho^x_{a} \rangle_{1+\epsilon} \\
&\le \left[\frac{1}{2} \left(\frac{\epsilon s\left(\ln2\right)}{q}\right)^2 +\frac{1}{6}\left(\frac{\epsilon s\left(\ln2\right)}{q}\right)^3 2^{\epsilon s/q}\right]\sum_{a,x} p(a)H(a,x)  \langle\rho^x_{a} \rangle_{1+\epsilon}  \\
&\le \frac{1}{2} \left(\frac{\epsilon s\left(\ln2\right)}{q}\right)^2 +\frac{1}{6}\left(\frac{\epsilon s\left(\ln2\right)}{q}\right)^3 2^{\epsilon s/q}
\end{aligned}
\end{equation}
After applying the function $-\frac{1}{\epsilon}\log()$, we have a more precise result similar to Proposition 6.3.
The difference is we replace the $O(\epsilon/q)$ by $ \frac{\epsilon}{q} \frac{\left(\ln2\right)s^2}{2} + (\frac{\epsilon}{q})^2 \frac{\left(\ln2\right)^2s^3}{6}2^{\epsilon s/q}$.
In the Theorem 6.7, we let the parameter $s$ be $\pi^{'}(\chi)$ .
In the Theorem 5.8, we know that
\begin{equation}
\begin{aligned}\label{rate6}
&\pi(\chi)=2\frac{\log(e)(\chi-w)^2}{r-1}\\
&\pi^{'}(\chi)=4\frac{\log(e)(\chi-w)}{r-1}
\end{aligned}
\end{equation}
Thus
\begin{equation}
\begin{aligned}\label{rate7}
O(\epsilon/q)=\frac{\epsilon}{q} \frac{8\log(e)(\chi-w)^2}{(r-1)^2} +  \left(\frac{\epsilon}{q}\right)^2 \frac{32 \log(e)(\chi-w)^3}{3(r-1)^3}      2^{\epsilon 4\frac{\log(e)(\chi-w)}{(r-1)q}}
\end{aligned}
\end{equation}
\textbf{Result 1}
\begin{equation}
\begin{aligned}\label{rate8}
&H_{min}^\delta(X|AE)\ge N [\pi(\chi)-\Delta]\\
&\pi(\chi)=2\frac{\log(e)(\chi-w)^2}{r-1}\\
&\Delta=\frac{\epsilon}{q} \frac{8\log(e)(\chi-w)^2}{(r-1)^2} +  \left(\frac{\epsilon}{q}\right)^2 \frac{32 \log(e)(\chi-w)^3}{3(r-1)^3}      2^{\frac{\epsilon}{q} \frac{4\log(e)(\chi-w)}{r-1}}+\frac{\log(2/\delta^2)}{N\epsilon}+2rq
\end{aligned}
\end{equation}
where $\chi\in[0,1]$ is the score obtained in experiments, $w$ is the classical bound for a certain game, $r$ is the number of total outputs, $q$ is the probability for test round, N is the total round number, $\delta$ is the failure probability, $\epsilon\in (0,1]$ is the .
The randomness expansion, generation, and input rate per round are
\begin{equation}
\begin{aligned}
R_{exp}&=R_{gen}-R_{In},\\
R_{gen}&=\pi(\chi)-\Delta,\\
R_{In}&=q\log11 +H(q).
\end{aligned}
\end{equation}

If we focus on the randomness expansion instead of the generation randomness, we should consider the random seed $H(q)+q\log 11$ used for random inputs. Different target function have different optimal result, the figures in main text shows the effect of optimization parameter.
Note that from the Result 1, the generated randomness is $O\left(N\right)$, and we take the probability $q\sim (\log^3 N)/N$, then the initial random seed required is $q\log11+H\left(q\right)$. And due to $\log N< N$, $q\log11+H\left(q\right)\sim O\left(q\right)+q \log \left(\left(\log^3N\right) /N\right)< O\left(\log^4 N\right)$. Thus compared with the generated randomness $O\left(N\right)$, exponential randomness expansion is achieved.

\section{Improved rate curve}
The important uncertain relation is related to the output alphabet size $r$. A larger $r$ will lead to a bad performance. This disadvantage is removed by an improved uncertain relation.
A tighter bound of Proposition 4.4 proposed by Ref. \cite{Huang17} is as follows.

\begin{lemma}
For any finite dimensional Hilbert space $V$ , any positive semidefinite operator $\tau: V \rightarrow V$, and any projective measurement $\{P_0, P_1,\cdots, P_n\}$ on $V$ , the following holds. Let $\tau^{'}=\sum_i P_i\tau P_i$. Then

\begin{equation}
\|\tau^{'}\|^2_{1+\epsilon}\le \|\tau\|^2_{1+\epsilon}-\epsilon \|\tau-\tau^{'}\|^2_{1+\epsilon}
\end{equation}
for all $\epsilon\in (0,1)$.
Consequently,
\begin{equation}
\|\tau^{'}\|_{1+\epsilon}\le \|\tau\|^2_{1+\epsilon}-\epsilon/2 \|\tau-\tau^{'}\|^2_{1+\epsilon}.
\end{equation}
\end{lemma}
This result can be applied in Theorem 5.8 and obtain a new rate curve,
\begin{equation}
\begin{aligned}\label{rate9}
\pi(\chi)=2\log(e)(\chi-w)^2~if~\chi \ge w.
\end{aligned}
\end{equation}
Consequently, we have $\pi^{'}(\chi)=4log(e)(\chi-w) $, and let the parameter $s$ be $\pi^{'}(\chi)$ in $O(\epsilon/q)$ by $ \frac{\epsilon}{q} \frac{\left(\ln2\right)s^2}{2} +  (\frac{\epsilon}{q})^2 \frac{\left(\ln2\right)^2s^3}{6}2^{\epsilon s/q}$.
Then
\begin{equation}
O(\epsilon/q)=\frac{\epsilon}{q} 8\log(e)(\chi-w)^2 + \left(\frac{\epsilon}{q}\right)^2 \frac{32 \log(e)(\chi-w)^3}{3}2^{\frac{\epsilon}{q} 4\log(e)(\chi-w)}.
\end{equation}

\textbf{Result 2}
\begin{equation}
\begin{aligned}\label{rate10}
&H_{min}^\delta(X|AE)\ge N [\pi(\chi)-\Delta]\\
&\pi(\chi)=2log(e)(\chi-w)^2\\
&\Delta=\frac{\epsilon}{q} 8log(e)(\chi-w)^2 +  \left(\frac{\epsilon}{q}\right)^2 \frac{32 log(e)(\chi-w)^3}{3}2^{\epsilon 4\frac{log(e)(\chi-w)}{q}}+\frac{log(2/\delta^2)}{N\epsilon}+2rq
\end{aligned}
\end{equation}

\section{Improvement of random number generation speed}
Currently, each round costs 3700 $\mu$s, which is consisted of 1500 $\mu$s cooling process, two detections procedures 900 $\mu$s in total, 140 $\mu$s spin echo pulses for the first detection, two optical pumping pulses 60 $\mu$s in total, rotations 60 $\mu$s in total, some short gaps between sequences to make sure they do not affect each other, and around 1000 $\mu$s communication time. However, there is room for technical improvement as follows. By extending coherence time between qutrit, spin echo will not be required. Detection time could be reduce to around 100 $\mu$s by replacing a high numerical aperture (NA) lens from 0.2 to 0.6. By amplifying 1762 $\mu$m laser power 10 times, Rabi oscillations between $\ket{1}$ and $\ket{3}$, and between $\ket{2}$ and $\ket{3}$ can be at least 3 times faster, so as the rotation. Each optical pumping could be reduced to 1 $\mu$s by further optimization. Currently we apply 1500 $\mu$s cooling process each round, but it will be possible to apply only one cooling process per ten rounds after some improvements. With all the development above, we can achieve at least one order faster generation speed.

\end{document}